\documentclass{llncs}
\usepackage{url}
\usepackage{epsf}
\usepackage{graphicx}
\usepackage{amsfonts}
\usepackage{amssymb}
\usepackage{amsmath}
\usepackage{latexsym}
\usepackage[table]{xcolor}

\newcommand{\ABox}{
\raisebox{3pt}{\framebox[6pt]{\rule{6pt}{0pt}}}
}

\newcommand{\e}{\varepsilon}

\newcommand{\id}{{\tt ID}}
\newcommand{\ssp}{{\tt sp}}
\newcommand{\disk}{{\tt disk}}

\newcommand{\alg}{{\sc LOS}}
\newcommand{\palg}{{\tt PLOS}}

\newcommand{\cc}{{\sc ClusterCover}}
\newcommand{\N}{L}  
\newcommand{\C}{C}
\newcommand{\udel}{{\tt UDel}}
\newcommand{\ydel}{{\tt YDel}}
\newcommand{\ldel}{{\tt LDel}}
\newcommand{\bps}{{\tt BPS}}

\newcommand{\lmst}{{\tt LMST}}
\newcommand{\pldel}{{\tt PLDel}}
\newcommand{\del}{{\tt Del}}
\newcommand{\oyao}{{\sc OrderedYao}}

\newcommand{\w}{\omega}

{\makeatletter
 \gdef\xxxmark{%
   \expandafter\ifx\csname @mpargs\endcsname\relax 
     \expandafter\ifx\csname @captype\endcsname\relax 
       \marginpar{xxx}
     \else
       xxx 
     \fi
   \else
     xxx 
   \fi}
 \gdef\xxx{\@ifnextchar[\xxx@lab\xxx@nolab}
 \long\gdef\xxx@lab[#1]#2{{\bf [\xxxmark #2 ---{\sc #1}]}}
 \long\gdef\xxx@nolab#1{{\bf [\xxxmark #1]}}
 \gdef\turnoffxxx{\long\gdef\xxx@lab[##1]##2{}\long\gdef\xxx@nolab##1{}}%
}

\title{\vspace{-2em}Localized Spanners for Wireless Networks}
\titlerunning{Spanners for UDGs and qUDGs}
\author{Mirela Damian\thanks{Supported by NSF grant CCF-0728909.}
\inst{1}\ \and Sriram V. Pemmaraju\inst{2}}
\institute{Dept. Comput. Sci., Villanova Univ., Villanova, PA
19085, USA. \email{mirela.damian@villanova.edu}.
\and
Dept. Comput. Sci., Univ. of Iowa, Iowa City, IA
52246, USA. \email{sriram@cs.uiowa.edu}.
}

\date{}

\begin{document}

\maketitle

\begin{abstract}
We present a new efficient localized algorithm to construct, for any
given quasi-unit disk graph $G=(V,E)$ and any $\e > 0$,
a $(1+\e)$-spanner for $G$ of maximum
degree $O(1)$ and total weight $O(\w(MST))$, where $\w(MST)$
denotes the weight of a minimum spanning tree for $V$. We further
show that similar localized techniques can be used to construct, for a given unit disk
graph $G = (V, E)$, a planar $C_{del}(1+\e)(1+\frac{\pi}{2})$-spanner
for $G$ of maximum degree $O(1)$ and total weight $O(\w(MST))$. Here $C_{del}$
denotes the stretch factor of the unit Delaunay triangulation for $V$.
Both constructions can be completed in $O(1)$ communication rounds,
and require each node to know its own coordinates.
\end{abstract}

\section{Introduction}
For any fixed $\alpha$, $0 < \alpha \le 1$, a graph $G = (V,E)$ is an
$\alpha$-\emph{quasi unit disk graph} ($\alpha$-QUDG) if there is an
embedding of $V$ in the Euclidean plane such that, for every vertex
pair $u, v \in V$, $uv \in E$ if $|uv| \le \alpha$, and $|uv|
\not\in E$ if $|uv| > 1$. The existence of edges with length in the
range $(\alpha, 1]$ is specified by an adversary. If $\alpha=1$, $G$
is called a \emph{unit disk graph} (UDG). $\alpha$-QUDGs
have been proposed as models for ad-hoc wireless networks
composed of homogeneous wireless nodes 
that communicate over a wireless medium without the aid of a fixed
infrastructure. 
%
Experimental studies show that the transmission range of a wireless
node is not perfectly circular and exhibits a transitional region
with highly unreliable links~\cite{ZK-qudg-07} (see for example
Fig.~\ref{fig:grid}a, in which the shaded region represents the
actual transmission range). In addition,
environmental conditions and physical obstructions adversely affect
signal propagation and ultimately the transmission range of a
wireless node. The parameter $\alpha$ in the $\alpha$-QUDG model
attempts to take into account such imperfections.

Wireless nodes are often powered by batteries and have limited
memory resources. These characteristics make it critical to
compute and maintain, at each node, only a subset of neighbors
that the node communicates with. This problem, referred to as
\emph{topology control}, seeks to adjust the transmission power
at each node so as to maintain connectivity, reduce
collisions and interference, and extend the battery lifetime
and consequently the network lifetime.


Different topologies optimize different performance metrics. In this
paper we focus on properties such as \emph{planarity}, \emph{low weight},
\emph{low degree}, and the \emph{spanner} property.
Another important property is
\emph{low interference}~\cite{BRWZ04,JC05,RSWZ05}, which we do not address
in this paper.
A graph is \emph{planar} if no two edges cross each other (i.e, no two edges
share a point other than an endpoint). Planarity is important to various
memoryless routing algorithms~\cite{KarpKung00,BMSU01}.
A graph is called \emph{low weight} if its total edge length,
defined as the sum of the lengths of all its edges, is within a constant
factor of the total edge length of the Minimum Spanning Tree (MST). It was
shown that the total energy consumed by sender nodes broadcasting along
the edges of a MST is within a constant factor of the
optimum~\cite{MSTBroadcast02}. 
\emph{Low degree} (bounded above by a constant) at each node is also
important for balancing
out the communication overhead among the wireless nodes. 
If too many edges are eliminated from the original graph however, paths
between pairs of nodes may become unacceptably long and offset the gain
of a low degree. This renders necessary a stronger requirement,
demanding that the reduced topology be a \emph{spanner}.
Intuitively, a structure is a spanner if it maintains
short paths between pairs of nodes in support of fast message delivery and
efficient routing. We define this formally below.

Let $G = (V, E)$ be a connected graph representing a wireless network.
For any pair of nodes $u,v \in V$, let $\ssp_G(u,v)$ denote 
a shortest path in $G$ from $u$ to $v$, and let $|\ssp_G(u,v)|$ denote
the length of this path. Let $H \subseteq G$ be a
connected subgraph of $G$. For fixed $t \ge 1$, $H$ is called a
$t$-\emph{spanner} for $G$ if, for all pairs
of vertices $u, v \in V$, $|\ssp_H(u,v)| \le t \cdot |\ssp_G(u,v)|$. The
value $t$ is called the \emph{stretch factor} of $H$. If $t$ is constant,
then $H$ is called a \emph{length spanner}, or simply a \emph{spanner}.
A triangulation of $V$ is a \emph{Delaunay triangulation},
denoted by \del($V$), if the
circumcircle of each of its triangles is empty of nodes in $V$.

Due to the limited resources and high mobility of the wireless nodes, it is
important to \emph{efficiently} construct and maintain a spanner in a \emph{localized} manner.
A \emph{localized} algorithm is a distributed
algorithm in which each node $u$ selects all its incident edges based
on the information from nodes within a constant number of hops from $u$.
Our communication model is the standard synchronous message passing model,
which ignores channel access and collision issues. In this communication model,
time is divided into \emph{rounds}. In a round, a node is able to receive all
messages sent in the previous round, execute local computations, and send messages
to neighbors. We measure the communication cost of our algorithms in terms of
\emph{rounds of communication}. The length of messages exchanged
between nodes is logarithmic in the number of nodes.

\paragraph{Our Results.}
In this paper we present the first localized method to construct, for any QUDG
$G = (V, E)$ and any $\e > 0$,
a $(1+\e)$-spanner for $G$ of maximum degree $O(1)$ and total weight $O(\w(MST))$, where
$\w(MST)$ denotes the weight of a minimum spanning tree for $V$. We further extend
our method to construct, for any UDG $G = (V, E)$,  a \emph{planar} spanner for $G$ of
maximum degree $O(1)$ and
total weight $O(\w(MST))$. The stretch factor of the spanner is bounded above by
$C_{del}(1+\e)(1+\frac{\pi}{2})$, where $C_{del}$ is the stretch factor of the
unit Delaunay triangulation for $V$ ($C_{del} \le 2.42$~\cite{LCW02}).
This second result resolves an open question
posed by Li et al. in~\cite{lws-ieee-04}.
Both constructions can be completed in $O(1)$ communication rounds,
and require each node to know its own coordinates.

\subsection{Related Work}
Several excellent surveys on spanners exist~\cite{Smid00,RajSurvey02,GK06,ns-gsn-07}.
In this section we restrict our attention to \emph{localized} methods for
constructing spanners for a given graph $G=(V,E)$. We proceed with a discussion
on non-planar structures for UDGs first. Existing results are summarized
in the first four rows of Table~\ref{tab:results}.

The \emph{Yao graph}~\cite{Yao82} with an integer parameter $k \ge 6$, denoted $YG_k$,
is defined as follows. At each node
$u \in V$, any $k$ equal-separated rays originated at $u$ define $k$ cones.
In each cone, pick a shortest edge $uv$, if there is any, and add to $YG_k$
the directed edge $\overrightarrow{uv}$. Ties are broken arbitrarily or by
smallest \id. The Yao graph is a spanner with stretch factor
$\frac{1}{1 - 2 \sin{\pi/k}}$, however its degree can be as high as $n-1$.
To overcome this shortcoming, Li et al.~\cite{li02sparse} proposed another
structure called \emph{YaoYao} graph $YY_k$, which is constructed by applying
a reverse Yao structure on $YG_k$: at each node $u$ in $YG_k$, discard all
directed edges $\overrightarrow{vu}$ from each cone centered at $u$, except for
a shortest one (again, ties can be broken arbitrarily or by smallest \id).
$YY_k$ has maximum node degree $2k$,
a constant. However, the tradeoff is unclear in that the question of
whether $YY_k$ is a spanner or not remains open. Both $YG_k$ and $YY_k$
have total weight $O(n)\cdot\w(MST)$~\cite{d-yys-08}.
Li et al.~\cite{WangLi03} further proposed another sparse structure,
called \emph{YaoSink} $YS_k$, that satisfies both the spanner and the
bounded degree properties. The sink technique replaces each directed star
in the Yao graph consisting of all links directed into
a node $u$, by a tree $T(u)$ with sink $u$ of bounded degree.
However, neither of these structures has low weight.

\begin{table}[htpb]
\begin{center}
\begin{tabular}{|l|c|c|c|c|c|c|} \hline
Structure & Planar? & Spanner? & Degree & Weight Factor & Comm. Rounds \\ \hline
{\tt YG$_k$}, $k \ge 6$~\cite{Yao82} & N & Y & $O(n)$ & $O(n)$  & $O(1)$ \\ \hline
{\tt YY$_k$}, $k \ge 6$~\cite{li02sparse} & N & ? & $O(1)$ & $O(n)$ & $O(1)$ \\ \hline
{\tt YS$_k$}, $k \ge 6$~\cite{WangLi03} & N & Y & $O(1)$ & $O(n)$ & $O(1)$ \\ \hline
\rowcolor{cyan}
{\tt LOS} [this paper] & N & Y & $O(1)$ & $O(1)$ & $O(1)$ \\ \hline
{\tt RDG}~\cite{GGH+01} & Y & Y & $O(n)$ & $O(n)$  & $O(1)$ \\ \hline
$\ldel^k$, $k \ge 2$~\cite{LCW02} & Y & Y & $O(n)$ & $O(n)$ & $O(1)$ \\ \hline
\pldel~\cite{LCW02,AraujoR04} & Y & Y & $O(n)$ & $O(n)$ & $O(1)$ \\ \hline
{\tt YaoGG}~\cite{li02sparse}& Y & N & O(n) & $O(n)$ & $O(1)$ \\ \hline
{\tt OrdYaoGG}~\cite{SWLF04}& Y & N & O(1) & $O(n)$ & $O(1)$ \\ \hline
\bps~\cite{WangLi03,LiWang04}& Y & Y & O(1) & $O(n)$ & $O(n)$ \\ \hline
{\tt RNG'}~\cite{li-localmst-03}& Y & N & O(1) & $O(1)$ & $O(1)$ \\ \hline
$\lmst_k, k \ge 2$~\cite{lws-ieee-04}& Y & N & O(1) & $O(1)$ & $O(1)$ \\ \hline
\rowcolor{cyan}
{\tt PLOS} [this paper] & Y & Y & O(1) & $O(1)$ & $O(1)$ \\ \hline
\end{tabular}
\end{center}
\caption[]{Results on localized methods for UDGs.}
\label{tab:results}
\end{table}

We now turn to discuss \emph{planar} structures for UDGs.
The relative neighborhood graph (RNG)~\cite{god-pr-80}
and the Gabriel graph (GG)~\cite{gg-geo-69} can both be constructed locally,
however neither is a spanner~\cite{BDEK06}.
On the other hand, the Delaunay triangulation \del($V$) is a planar $t$-spanner
of the complete Euclidean graph with vertex set $V$.
This result was first proved by Dobkin, Friedman and Supowit~\cite{dfs-dg-90},
for $t = \frac{1 + \sqrt{5}}{2}\pi \approx 5.08$, and was further improved
to $t = \frac{4 \sqrt{3}}{9} \pi \approx 2.42$ by Keil and Gutwin~\cite{kg-dt-92}.
Das and Joseph~\cite{dj-tcg-89} generalize these results by identifying
two properties of planar graphs, the good polygon and diamond properties,
which imply that the stretch factor is bounded above by a constant.



For a given point set $V$, the unit Delaunay triangulation of $V$, denoted \udel($V$),
is the graph obtained by removing all Delaunay edges from \del($V$) that are
longer than one unit.
It was shown that \udel($V$) is a $t$-spanner of the unit-disk graph
UDG($V$), with $t = \frac{4\sqrt{3}}{9}\pi \approx 2.42$~\cite{LCW02}.

Gao et al.~\cite{GGH+01} present a localized algorithm to build a
planar spanner called restricted Delaunay graph ({\tt RDG}), which is a
supergraph of \udel($V$).
Li et al.~\cite{LCW02} introduce the notion of a
$k$-\emph{localized Delaunay triangle}: $\triangle abc$ is called
$k$-\emph{localized Delaunay} if the interior of its circumcircle does not
contain any node in $V$ that is a $k$-neighbor of $a$, $b$ or $c$,
and all edges of $\triangle abc$ are no longer than
one unit. The authors describe a localized method to construct, for
fixed $k \ge 1$, the $k$-localized Delaunay graph $\ldel^k(V)$, which
contains all Gabriel edges and edges of all $k$-localized
Delaunay triangles. They show that (i) $\ldel^k(V)$ is a supergraph
of $\udel(V)$ (and therefore a $\frac{4\sqrt{3}}{9}\pi$-spanner), (ii)
$\ldel^k(V)$ is planar, for any $k \ge 2$, and (iii) $\ldel^1(V)$ may
not be planar, but a planar subgraph $\pldel(V) \subseteq  \ldel^1(V)$ that
retains the spanner property can be locally extracted from $\ldel^1(V)$.
Their planar spanner constructions take 4 rounds of communication and a
total of $O(n)$ messages ($O(n \log n)$ bits).
Ara{\'u}jo and Rodrigues~\cite{AraujoR04} improve upon the communication
time for $\pldel$ and devise a method to compute $\pldel(V)$ in one
single communication step. Both $\pldel(V)$ and $\ldel^k(V)$, for $k \ge 1$,
may have arbitrarily large degree and weight.

To bound the degree, several methods apply the \emph{ordered Yao} structure
on top of an unbounded-degree planar structure. This idea was first introduced by
Bose et al. in~\cite{BGMS02},
and later refined by Li and Wang in~\cite{WangLi03,LiWang04}. Since the
ordered Yao structure is relevant to our work in this paper as well,
we pause to discuss the \oyao\ method for constructing this structure.
The \oyao\ method is outlined in Table~\ref{tab:orderedyao}. The main idea
is to define an ordering $\pi$ of the nodes such that each node $u$ has
a limited number of neighbors (at most 5) who are predecessors in $\pi$;
these predecessors are used to define a small number of open cones
centered at $u$, each of which will contain at most one neighbor of $u$
in the final structure. To maintain the spanner property of the original graph,
a short path connecting all neighbors of $u$ in each cone is used
to replace the edges incident to $u$ that get discarded from the original graph.

%

\noindent
Thm.~\ref{thm:oyao} summarizes the important properties of the structure computed by the \oyao\ method.
\begin{theorem}
If $G$ is a planar graph, then the output $G'$ obtained by executing \oyao$(G)$ is a planar $(1+\frac{\pi}{2})$-spanner for $G$ of maximum degree 25~\emph{\cite{WangLi03}}.
\label{thm:oyao}
\end{theorem}

\begin{table}[htpb]
\begin{center}
\fbox{
\begin{tabular}{lc}
\multicolumn{2}{c}
{\begin{minipage}[ht]{0.8\linewidth}
\centerline{Algorithm \oyao($G=(V,E)$)~\cite{WangLi03}}
\vspace{1mm}{\hrule width\linewidth}\vspace{2mm} 
\end{minipage}}
\\
{\begin{minipage}[ht]{0.45\linewidth}
\vspace{-5em}\small{\begin{tabbing}
.......\=......\=.......\=.........\=..........\=.....................\=..............................\kill
{\bf \{1. Find an order $\pi$ for $V$:\}} \\
\> Initialize $i = 1$ and $G_i = G$. \\
\> Repeat for $i = 1, 2, \ldots, |V|$ \\
\> \> Remove from $G_i$ the node $u$ of smallest degree \\
\> \> (break ties by smallest \id.) \\
\> \> Call the remaining graph $G_{i+1}$. \\
\> \> Set $\pi_u = n - i + 1$.
\end{tabbing}}
\end{minipage}}
 &
\includegraphics[width=0.22\linewidth]{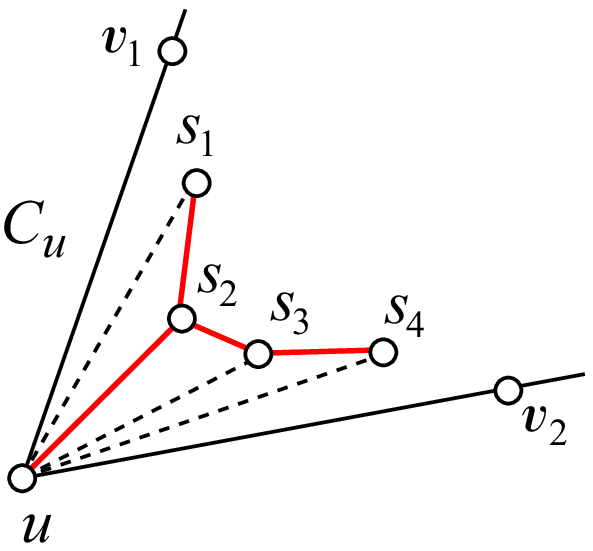} \\
\\
\multicolumn{2}{l}{\begin{minipage}[ht]{0.8\linewidth}
\small{\begin{tabbing}
.......\=......\=.......\=.........\=..........\=.....................\=..............................\kill
{\bf \{2. Construct a bounded-degree structure for $G$:\}} \\
\> Mark all nodes in $V$ \emph{unprocessed}. Initialize $E' \leftarrow \emptyset$ and $G' = (V, E')$.\\
\> Repeat $|V|$ times \\
\> \> Let $u$ be the unprocessed node with the smallest order $\pi_u$. \\
\> \> Let $v_1, v_2, \ldots, v_h$ be the be the processed neighbors of $u$ in $G$ ($h \le 5$). \\
\> \> Shoot rays from $u$ through each $v_i$, to define $h$ sectors centered at $u$. \\
\> \> Divide each sector into fewest open cones of degree at most $\pi/3$. \\
\> \> For each such open cone $C_u$ (refer to Fig. above) \\
\> \> \> Let $s_1, s_2, \ldots, s_m$ be the geometrically ordered neighbors of $u$ in $C_u$. \\
\> \> \> Add to $E'$ the shortest $us_i$ edge. \\
\> \> \> Add to $E'$ all edges $s_js_{j+1}$, for $j = 1, 2, \ldots, m-1$. \\
\> \> Mark node $u$ \emph{processed}. \\
\\
{\bf Output $G' = (V, E')$.}
\end{tabbing}}
\end{minipage}}
\end{tabular}
}
\end{center}
\caption{The \oyao\ method.}
\label{tab:orderedyao}
\end{table}

Song et al.~\cite{SWLF04} apply the ordered Yao structure
on top of the Gabriel graph {\tt GG}($V$) to produce
a planar bounded-degree structure {\tt OrdYaoGG}. Their result improves
upon the earlier localized structure
{\tt YaoGG}~\cite{li02sparse}, which may not have bounded degree. Both
{\tt YaoGG} and {\tt OrdYaoGG} are power spanners, 
however neither is a length spanner. 

The first efficient localized method to construct a bounded-degree
planar spanner was proposed by Li and Wang in~\cite{WangLi03,LiWang04}.
Their method applies the ordered Yao structure on top of $\ldel(V)$ to
bound the node degree. The resulted structure, called $\bps(V)$
(Bounded-Degree Planar Spanner), has degree bounded above by
$19 + \lceil \frac{2\pi}{\alpha} \rceil$,
where $0 < \alpha < \frac{\pi}{3}$ is an adjustable parameter.
The total communication complexity for constructing $\bps(V)$ is
$O(n)$ messages, however it may take as many as $O(n)$ rounds of
communication for a node to find its rank in the ordering of $V$
(a trivial example would be $n$ nodes
lined up in increasing order by their \id).
The \bps\ structure does not have low weight \cite{li-localmst-03}.

The first \emph{localized low-weight} planar structure was proposed
in~\cite{li-localmst-03}. This structure, called {\tt RNG'},
is based on a modified relative neighborhood graph, and satisfies the planarity,
bounded-degree and bounded-weight properties.
A similar result has been obtained by Li, Wang and Song~\cite{lws-ieee-04},
who propose a family of structures, called \emph{Localized Minimum Spanning Trees}
$\lmst_k$, for $k \ge 1$. The authors show that $\lmst_k$ is planar, has
maximum degree 6 and total weight within a constant factor of $\w(MST)$,
for $k \ge 2$. Their result extends an earlier result by Li, Hou and
Sha~\cite{lhs-infocom-03}, who propose a localized MST-based method to
compute a local minimum spanning tree structure.
However, neither of these low-weight structures satisfies the spanner property.
Constructing low-weight, low-degree planar spanners in few rounds of
communication is one of the open problems we resolve in this paper.

\section{Our Work}
We start with a few definitions and notation to be used through
the rest of the paper. For any nodes $u$ and $v$, let $uv$ denote the
edge with endpoints $u$ and $v$;
$\overrightarrow{uv}$ is the edge directed from $u$ to $v$; and
$|uv|$ denotes the Euclidean distance between $u$ and $v$.
Let $\C_u$ denote an arbitrary cone with
apex $u$, and let $\C_u(v)$ denote the cone with apex $u$ containing $v$.
For any edge set $E$ and any cone $\C_u$, let $E \cap \C_u$ denote the
subset of edges in $E$ incident to $u$ that lie in $\C_u$.

We assume that each node $u$ has a unique identifier \id($u$) and knows
its coordinates $(x_u, y_u)$.
Define the identifier $\id(\overrightarrow{uv})$ of a directed edge
$\overrightarrow{uv}$ to be the triplet $(|uv|, \id(u), \id(v))$.
For any pair of directed edges $\overrightarrow{uv}$ and
$\overrightarrow{u'v'}$, we say that $\id(\overrightarrow{uv}) <
\id(\overrightarrow{u'v'})$ if and only if one of the following
conditions holds: (1) $|uv| < |u'v'|$, or (2)
$|uv| = |u'v'|$ and $\id(u) < \id(u')$, or (3)
$|uv| = |u'v'|$ and $\id(u) = \id(u')$ and $\id(v) < \id(v')$.
For an undirected edge $uv$, define $\id(uv) =
\min\{\id(\overrightarrow{uv}), \id(\overrightarrow{vu})\}$. Note that according
to this definition, each edge has a unique identifier.

Let $H = (V, E_H)$ be an arbitrary subgraph of $G = (V, E)$. A subset
$\N_u \subset V$ is an $r$-\emph{cluster} in $H$ with center $u$ if, for
any $v \in \N_u$, $|\ssp_H(u, v)| \le r$. A set of disjoint
$r$-clusters $\{\N_{u_1}, \N_{u_2}, \ldots\}$ form an $r$-\emph{cluster cover} for
$V$ in $H$ if they satisfy two properties: (i) for $i \neq j$, $|\ssp_H(u_i,
u_j)| > r$ (the $r$-\emph{packing} property), and (ii) the union
$\cup_{i}\N_{u_i}$ covers $V$ (the $r$-\emph{covering} property).

For any node subset $U \subseteq V$, let $G[U]$ denote the subgraph of $G$ induced by $U$.
A set of node subsets $V_1, V_2, \ldots \subseteq V$ is a \emph{clique cover} for $V$ if
the subgraph of $G[V_i]$ is a clique for each
$i$, and $\cup_{i=1}^h V_i = V$.

The \emph{aspect ratio} of an edge set $E$ is the ratio of the
length of a longest edge in $E$ to the length of a shortest edge in
$E$. The aspect ratio of a graph is defined as the aspect ratio of
its edge set.

\subsection{The \alg\ Algorithm}
\label{sec:alg1}
%
%
In this section we describe an algorithm called \alg(Localized Optimal Spanner)
that takes as input an $\alpha$-QUDG $G = (V, E)$, for fixed $0 < \alpha \le 1$, and a value
$\e > 0$, and computes a $(1+\e)$-spanner for $G$ of maximum degree $O(1)$ and
total weight $O(\w(MST))$.
The main idea of our algorithm is to
compute a particular clique cover $V_1, V_2, \ldots$ for $V$, construct a
$(1+\e)$-spanner for each $G[V_i]$, then connect these smaller spanners
into a $(1+\e)$-spanner for $G$ using selected Yao edges.
In the following we discuss the details of our algorithm.

\begin{figure}[htbp]
\centering
\begin{tabular}{c@{\hspace{0.02\linewidth}}c@{\hspace{0.02\linewidth}}c}
\includegraphics[width=0.32\linewidth]{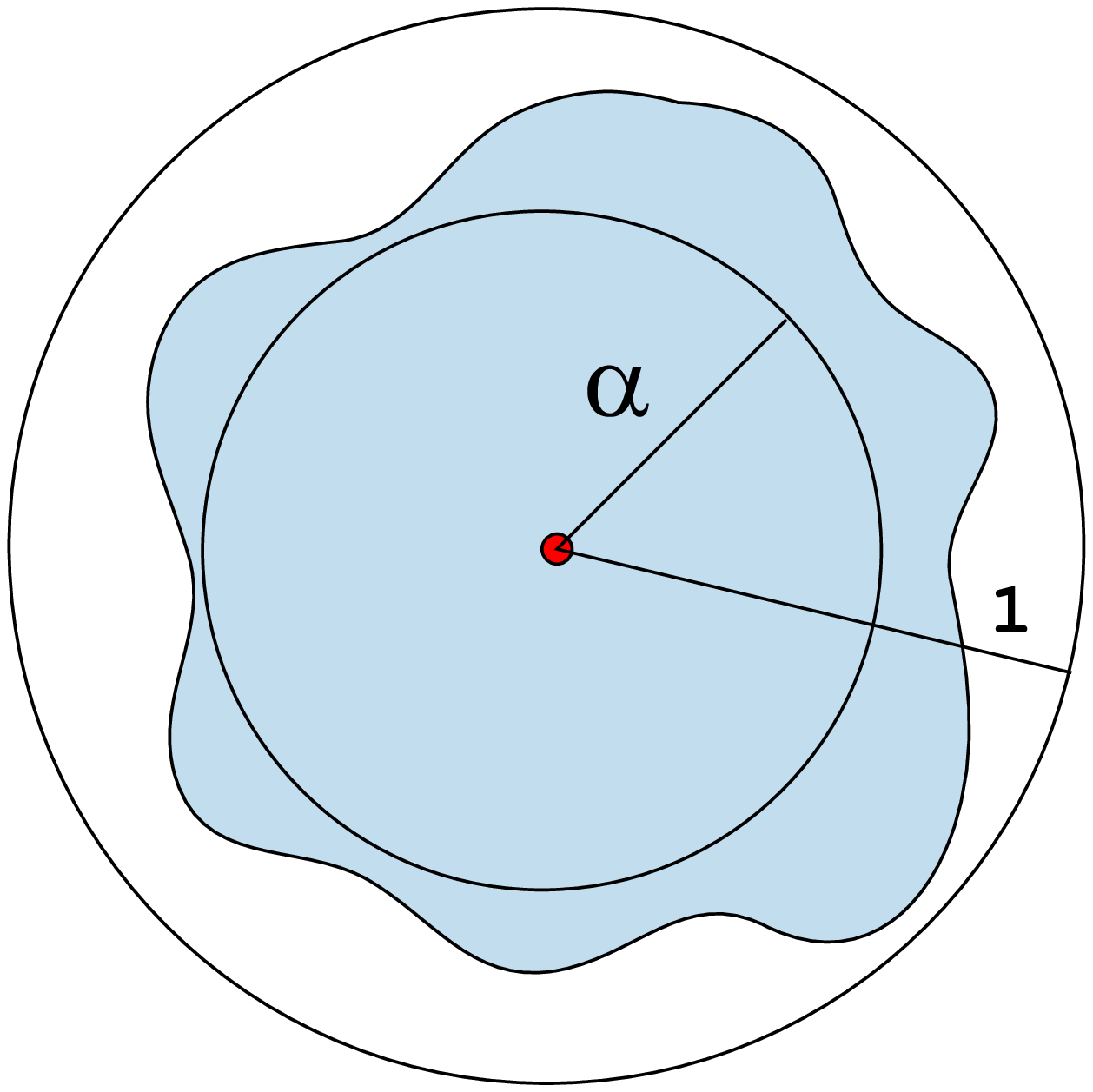} &
\includegraphics[width=0.32\linewidth]{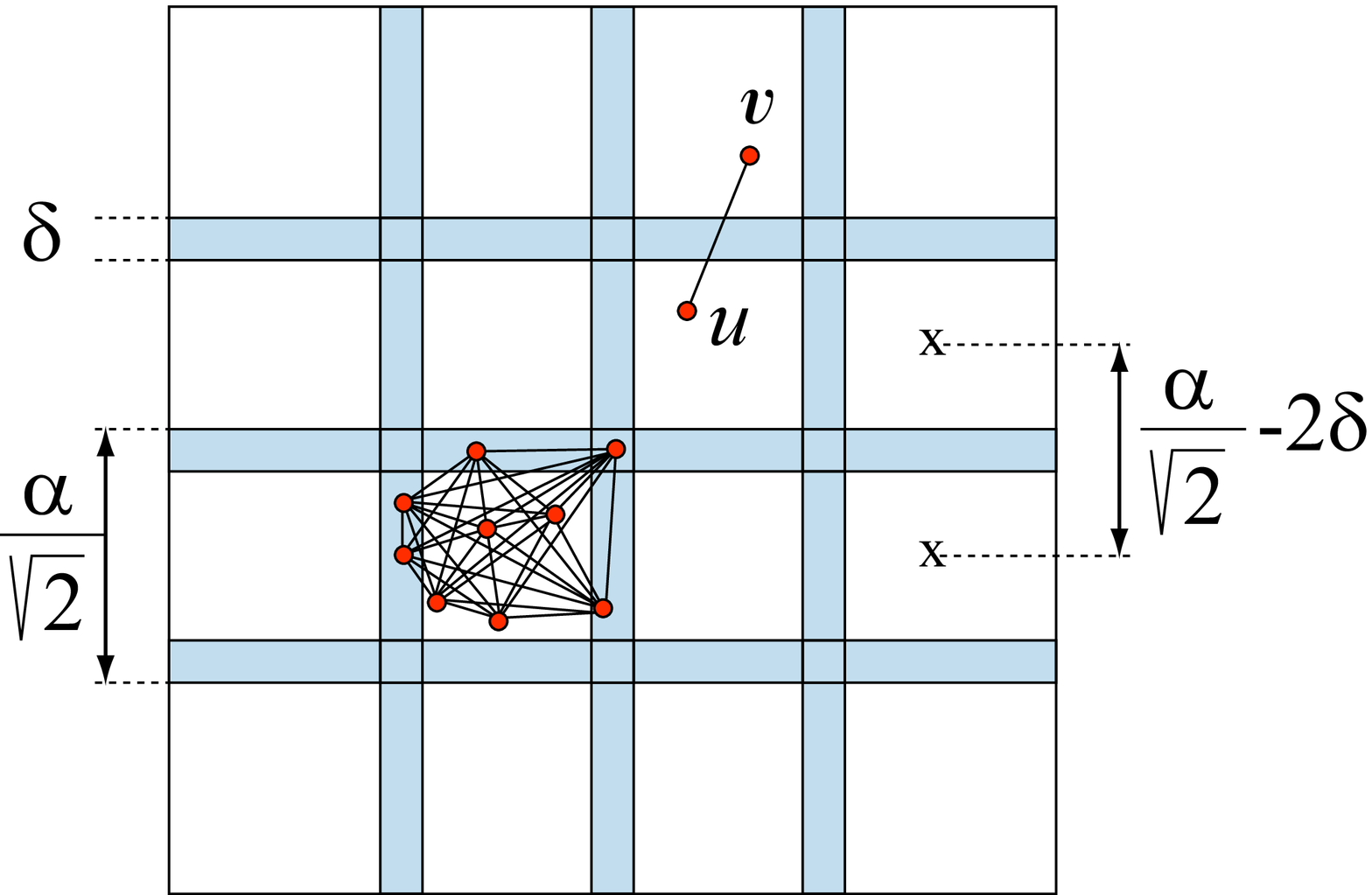} &
\includegraphics[width=0.21\linewidth]{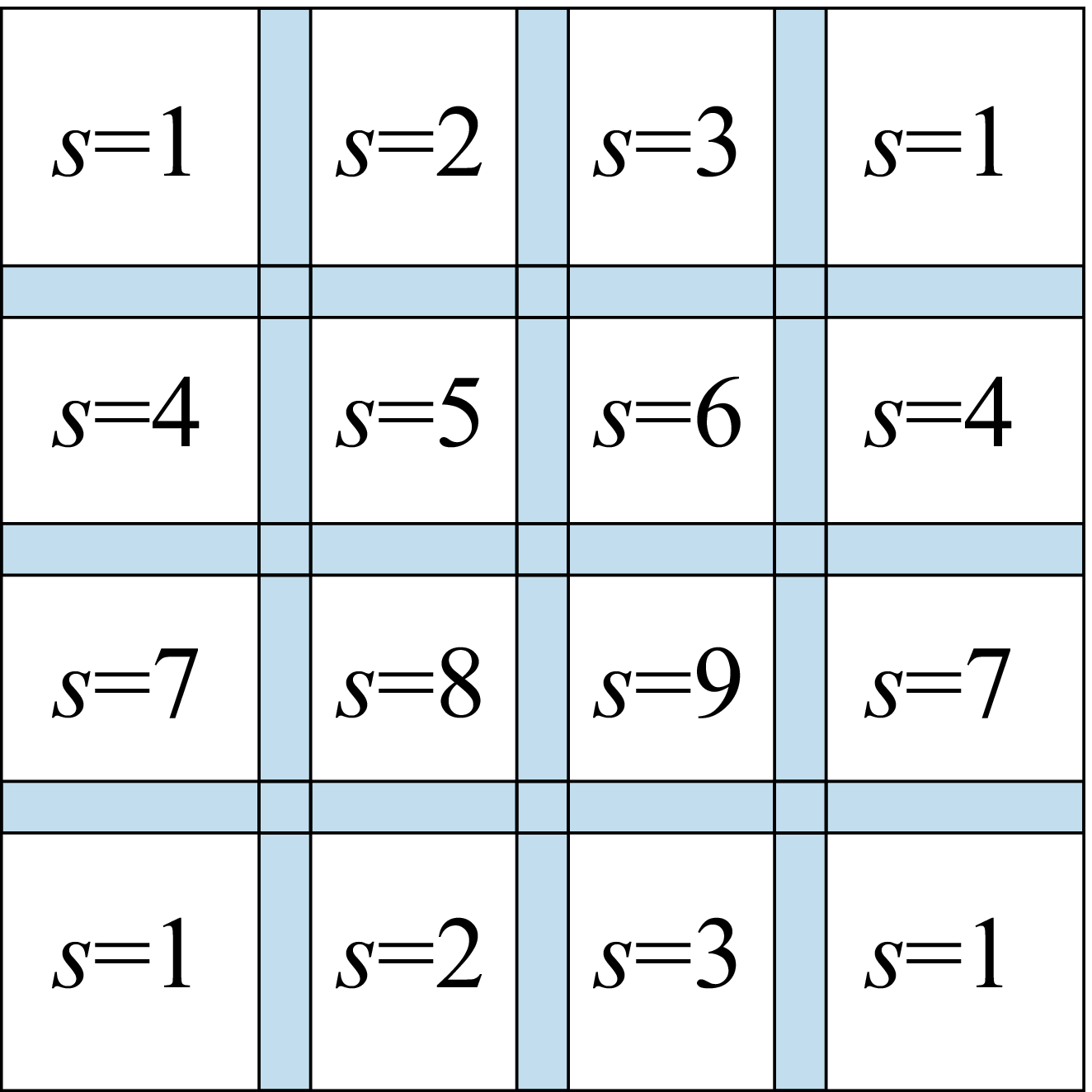} \\
(a) & (b) & (c)
\end{tabular}
\caption{(a) The $\alpha$-QUDG model (b) Constructing a clique cover for $V$ (c) Clique ordering.}
\label{fig:grid}
\end{figure}

Let $0 < \beta < \frac{\alpha}{\sqrt{2}}$ and
$0 < \delta < \beta/4$ be small constants to be fixed later.
To compute a clique cover for $V$, we
start by covering the plane with a grid of overlapping square cells of size
$\beta \times \beta$, such that the
distance between centers of adjacent cells is
$\beta - 2\delta$. Note that any two adjacent
cells define a small band of width $\delta$ where they overlap. The
reason for enforcing this overlap is to ensure that edges
not entirely contained within a single grid cell are longer
than $\delta$, i.e., they cannot be arbitrarily small.
We identify each grid cell by the coordinates $(i, j)$ of its upper left corner.
Any two vertices that lie within the same grid cell are no more than
$\alpha$ distance apart and therefore are connected by an edge in $G$.
This implies that the collection of vertices in each non-empty grid
cell can be used to define a clique element of the clique cover.
We call this
particular clique cover a ($\beta,\delta$)-\emph{clique cover}.
Let $V_1, V_2, \ldots $ be the elements of the ($\beta,\delta)$-clique cover for $V$.
Note that, since $\delta < \beta/4$, a node $u$ can
belong to at most four subsets $V_i$.

Our \alg\ method consists of 4 steps. First we construct, for each
$G[V_i]$, a $(1+\e)$-spanner  of degree $O(1)$ and weight
$O(\w(MST(V_i))$.
Various methods for constructing
$H_i$ exist -- for instance, the well-known sequential greedy
method produces a spanner with the desired properties~\cite{DN97}.
Second, we use the
Yao method to generate $(1+\e)$-spanner paths between longer edges that
span different grid cells. Third, we apply the reverse Yao step
to reduce the number of Yao edges incident to each node.
Finally, we apply a filtering method to eliminate all but a
constant number of edges incident to a grid cell. This fourth step
is necessary to ensure that the output spanner has bounded weight.
These steps are described in detail in Table~\ref{tab:alg}.
\begin{table}[ht]
\begin{center}
\fbox{
\begin{minipage}[h]{0.8\linewidth}
\centerline{Algorithm \alg($G=(V,E), \e$)}
\vspace{1mm}{\hrule width\linewidth}\vspace{2mm} 
\small{\begin{tabbing}
.......\=......\=.......\=.......\=..................................................\kill
{\bf \{1. Compute a $(1+\e)$-spanner cover:\}} \\
\> Fix $0 < \beta < \frac{\alpha}{\sqrt{2}}$ and $0 < \delta < \beta/4$. \\
\> Compute a ($\beta,\delta$)-clique cover $V_1, V_2, \ldots$ for $V$. \\
\> For each $i$, compute a $(1+\e)$-spanner $H_i$ for $G[V_i]$ using the method from~\cite{DN97}. \\
\> Initialize $H = \cup_i H_i$. Let \mbox{$E_0 = \{ uv \in E ~|~ uv \not\in G[V_i]$ for any $i$\}.} \\
\\
{\bf \{2. Apply Yao on $E_0$:\}} \\
\> Let $k$ be the smallest integer satisfying $\cos\frac{2\pi}{k} - \sin\frac{2\pi}{k} \ge \frac{\delta+1+\e}{(\delta+1)(1+\e)}$.\\
\> For each node $u$, divide the plane into $k$ incident equal-size cones. \\
\> Initialize $E_Y \leftarrow \emptyset$. \\
\> For each cone $\C_u$ such that $E_0 \cap \C_u$ is non-empty \\
\> \> Pick the edge $uv \in E_0 \cap \C_u$ of smallest \id\ and
add $\overrightarrow{uv}$ to $E_Y$.\\
\\
{\bf \{3. Apply reverse Yao on $E_Y$:\}} \\ 
   \> Initialize $E_{YY} \leftarrow E_{Y}$. \\
   \> For each cone $\C_u$ such that $E_{Y} \cap \C_u$ is non-empty \\
   \> \> Discard from $E_Y$ all edges $\overrightarrow{vu} \in E_Y \cap \C_u$, but the one of smallest \id.\\
\\
{\bf \{4. Select connecting edges from $E_{YY}$:\}} \\
\> Pick $r$ such that $r \le \frac{(\delta+1)(1+\e)(\cos\theta-\sin\theta)-(\delta+1+\e)}{4}$, where  $\theta  =2\pi/k$. \\
\> Compute an $r$-cluster cover for $V$ in $H$. \\
\> Let $E_1 \subseteq E_{YY}$ contain all Yao edges connecting cluster centers. Add $E_1$ to $H$. \\
\\
{\bf Output $H = (V, E_H)$.}
\end{tabbing}}
\end{minipage}
}
\end{center}
\caption{The \alg\ algorithm.}
\label{tab:alg}
\end{table}
Note that the Yao and reverse Yao steps are restricted to edges
in the set $E_0$ whose aspect ratio is bounded above by $1/\delta$.
The next three theorems prove the main properties of the \alg\ algorithm.

\begin{theorem}
The output $H$ generated by \alg($G, \e$) is a $(1+\e)$-spanner for $G$.
\label{thm:spanner}
\end{theorem}
\begin{proof}
Let $uv \in E$ be arbitrary. If $uv \in G[V_i]$ for some $i$, then $H_i \subseteq H$ contains
a $(1+\e)$-spanner $uv$-path (since $H_i$ is a $(1+\e)$-spanner for $G[V_i]$). Otherwise, $uv \in E_0$.
The proof that $H$ contains a $(1+\e)$-spanner $uv$-path is by induction on the \id\ of
edges in $E_0$. Let $uv \in E_0$ be the edge with the smallest \id\ and assume without loss of generality
that $\id(uv) = \id(\overrightarrow{uv})$. Since $\id(uv)$ is
smallest, $\overrightarrow{uv}$ gets added to $E_Y$ in step 2, and it stays in
$E_{YY}$ in step 3. If $uv \in H$ at the end of step 4, then $\ssp_H(u,v) = uv$.
Otherwise, let $ab$ be the edge selected in step 4 of the algorithm,
such that $u \in \N_{a}$ and $v \in \N_{b}$ (see Fig.~\ref{fig:spanner}a).
Since $\N_a$ and $\N_b$ are both $r$-clusters, we have that
$|\ssp_{H}(u, a)| \le r$ and $|\ssp_H(v, b)| \le r$. It follows that $|ua| \le r$ and $|vb| \le r$.
By the triangle
inequality, $|ab| < |uv| + 2r$ and therefore $\ssp_H(u,v) \le |ab| + 2r < |uv| + 4r \le (1+\e)|uv|$, for
any  $r \le \delta\e/4$ (satisfied by the $r$ values restricted by the algorithm).
This concludes the base case.

To prove the inductive step, let $uv \in E_0$ be arbitrary, and assume that $H$ contains
$(1+\e)$-spanner paths between the endpoints of any edge whose \id\ is lower than \id($uv$).

\begin{figure}[ht]
\centering
\begin{tabular}{c@{\hspace{0.1\linewidth}}c}
\includegraphics[width=0.35\linewidth]{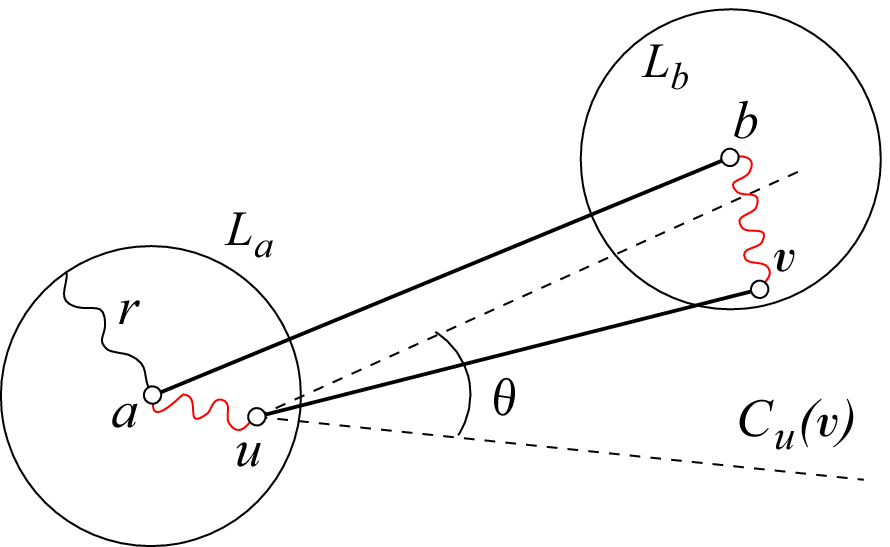} &
\includegraphics[width=0.35\linewidth]{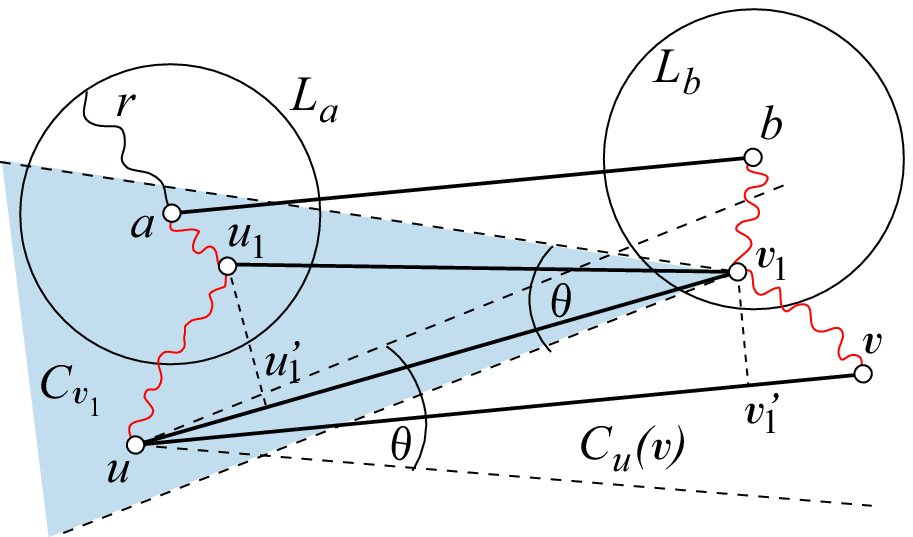} \\
(a) & (b)
\end{tabular}
\caption{Thm.~\ref{thm:spanner}: (a) Base case. (b)
$\ssp_H(u,u_1) \oplus \ssp_H(u_1, a) \oplus ab \oplus \ssp_H(b, v_1) \oplus \ssp_H(v_1, v)$ is a $(1+\e)$-spanner $uv$-path.}
\label{fig:spanner}
\end{figure}

Let $uv_1 \in \C_u(v)$ be the Yao edge selected in step 2 of the algorithm;
let $u_1v_1 \in C_{v_1}(u)$ be the YaoYao edge selected in step 3 of the algorithm;
and let $ab \in H$ be the edge added to $H$ in step 4 of the algorithm, such that
$u_1 \in \N_a$ and $v_1 \in \N_b$ (see Fig.~\ref{fig:spanner}b).
Note that $u$ and $u_1$ may be disjoint or may coincide, and similarly
for $v$ and $v_1$. In either case, the chains of inequalities
$\id(u_1v_1) \le \id(uv_1) \le \id(uv)$ and $|u_1v_1| \le |uv_1| \le |uv|$ hold.
Let $u'_1$ be the projection of $u_1$ on $uv_1$.
By the triangle inequality,
\begin{equation}
|uu_1| \le |uu'_1|+|u'_1u_1| = |uv_1| - |u'_1v_1| + |u'_1u_1| \le
|uv_1| - |u_1v_1|\cos\theta + |u_1v_1|\sin\theta. \label{eq:q1}
\end{equation}
Similarly, if $v'_1$ is the projection of $v_1$ on $uv$, we have
\begin{equation}
|v_1v| \le |vv'_1|+|v'_1v_1| = |uv| - |uv'_1| + |v'_1v_1| \le |uv| -
|uv_1|\cos\theta + |uv_1|\sin\theta. \label{eq:q2}
\end{equation}
Since $|uu_1| < |uv_1| \le |uv|$ and $|v_1v| < |uv|$, by the inductive
hypothesis $H$ contains $(1+\e)$-spanner paths $\ssp_H(u,u_1)$ and $\ssp_H(v_1,v)$.
Let $P_1 = \ssp_H(u,u_1) \oplus \ssp_H(v_1, v)$.
The length of $P_1$ is
$$
|P_1| \le (1+\e)\cdot(|uu_1| + |v_1v|).
$$
Substituting inequalities~(\ref{eq:q1}) and~(\ref{eq:q2}) yields
\begin{equation}
|P_1| \le (1+\e)|uv| + (1+\e)|uv_1|(1-\cos\theta +\sin\theta) - (1+\e)|u_1v_1|(\cos\theta  - \sin\theta).
\label{eq:q3}
\end{equation}
Next we show that the path $P = P_1 \oplus \ssp_H(u_1, a) \oplus ab \oplus \ssp_H(b, v_1)$ 
is a $(1+\e)$-spanner path from $u$ to $v$ in $H$, thus proving the inductive step. Using the fact that $|ab| < 2r + |u_1v_1|$, $|\ssp_H(u_1,a)| \le r$ and $|\ssp_H(b, v_1)| \le r$, we get
%
\begin{equation}
|P| \le |P_1| + |u_1v_1| + 4r.
\label{eq:q4}
\end{equation}
Substituting further $|u_1v_1| \ge \delta$ and $|uv_1| \le 1$ in~(\ref{eq:q3}) and~(\ref{eq:q4}) yields
$$
|P| \le (1+\e) |uv| + (4r + (1+\e)(1-\cos\theta +\sin\theta) - \delta(1+\e)(\cos\theta  - \sin\theta)-\delta).
$$
Note that the second term on the right side of the inequality above is
non-positive for any $r$ and $\theta$ satisfying the conditions of the algorithm:
\[
\begin{cases}
r \le \frac{(\delta+1)(1+\e)(\cos\theta - \sin\theta) - (\delta+1+\e)}{4} \\
\cos\theta - \sin\theta > \frac{\delta+1+\e}{(\delta+1)(1+\e)}.
\end{cases}
\]
This completes the proof. \hfill\ABox
\end{proof}

Before proving the other two properties of $H$ (bounded degree and bounded weight), we
introduce an intermediate lemma. For fixed $c > 0$, call an edge set
$F$ \emph{$c$-isolated} if,
for each node $u$ incident to an edge $e \in F$, the closed disk $\disk(u,c)$ centered at
$u$ of radius $c$ contains no other endpoints of edges in $F$. This definition is a
variant of the \emph{isolation property} introduced in~\cite{Das95}. Das et al. show
that, if an edge set $F$ satisfies the isolation property, then $\w(F)$ is within a
constant factor of the minimum spanning tree connecting the endpoints of $F$. Here we
prove a similar result.

\begin{lemma}
Let $F$ be a $c$-isolated set of edges no longer than 1.
Then $\w(F) = O(1)\cdot\w(T)$, where $T$ is the
minimum spanning tree connecting the endpoints of edges in $F$.
\label{lem:isolated}
\end{lemma}
\begin{proof}
Let $P$ be a Hamiltonian path obtained by a taking a preorder traversal of $T$.
If each edge $uv \in P$ gets associated a weight value $\w(uv) = |\ssp_T(u,v)|$,
then it is well-known that $\w(P) \le 2\w(T)$. So in order to prove that $w(F)$
is within a constant factor of $\w(T)$, it suffices to show that $\w(F) = O(\w(P))$.
Since $F$ is $c$-isolated, the distance between any two vertices in $T$ is greater
than $c$ and therefore $w(P) \ge (n-1)c$.
On the other hand, no edge in $F$ is greater than 1 and therefore
$\w(F) \le n$. It follows that $\w(F) = O(\w(P))$. \hfill\ABox
\end{proof}

\begin{theorem}
The output $H$ generated by running \alg($G, t$) has maximum degree $O(1)$ and total weight $O(1) \cdot \w(MST)$.
\label{thm:degreeweight}
\end{theorem}
\begin{proof}
The fact that $H$ has maximum degree $O(1)$ follows immediately from three observations:
(a) each spanner $H_i$ constructed in step 1 of the algorithm has degree $O(1)$~\cite{DN97},
(b) a node $u$ belongs to at most four subgraphs $H_i$, and
(c) a node $u$ is incident to a constant number of Yao edges (at most $2k$)~\cite{li02sparse}.

We now prove that the total weight for $H$ is within a constant factor of $\w(MST)$,
which is optimal. The main idea is to partition the edge set $E_H$ into a constant
number of subsets, each of which has low weight.
Consider first the $(1+\e)$-spanners constructed in step 1 of the algorithm.
Each $(1+\e)$-spanner $H_\ell$ corresponds to a grid cell $(i, j)$. Let $F$ denote
the set of edges in $\cup_\ell H_\ell$. Define the edge
set $F_s \subseteq F$ to contain all spanner edges corresponding to those grid cells $(i, j)$
whose indices $i$ and $j$ satisfy the condition $(i \mod 3) \times 3+ j \mod 3 = s$.
Intuitively, if two edges $e_1, e_2 \in F_{s}$ lie in different grid cells, then
those grid cells are separated by at least two other grid cells (see Fig.~\ref{fig:grid}c).
This further implies that the closest endpoints of $e_1$ and $e_2$ are distance $\alpha$
or more apart. Also notice that it takes only 9 subsets $F_1, F_2, \ldots, F_9$ to cover
$F$.

Next we show that $\w(F_s) =O(\w(T_s))$ for each $s = 1,2,\ldots, 9$, where
$T_{s}$ is a minimum spanning tree connecting the endpoints in $F_s$.
To see this, first observe that $F_s$ combines the edges of several low-weight $(1+\e)$-spanners
$H_{s_1}, H_{s_2}, \ldots$ with the property that $\w(H_{s_1}) = O(\w(T_{s_1}))$,
where $T_{s_1}$ is a minimum spanning tree connecting the nodes in $H_{s_1}$. Thus, in order
to prove that $\w(F_s) = O(\w(T_s))$, it suffices to show
$\sum_i \w(T_{s_i}) = O(\w(T_s))$. We will in fact prove that
\[
\sum_i \w(T_{s_i}) \le \w(T_s)
\]
We prove this by showing that, if Prim's algorithm is employed in constructing $T_s$ and $T_{s_i}$, then
$T_{s_i} \subseteq T_s$, for each $i$. Since the trees $T_{s_i}$ are all disjoint (separated by at
least 2 grid cells), the claim follows.
Recall that Prim's algorithm processes edges by increasing length and adds them to
$T_s$ as long as they do not close a cycle. This means that all edges shorter than
$\alpha$ are processed before edges longer than $\alpha$. Let $e \in T_{s_i}$ be
arbitrary. Then $|e| \le \alpha$, since $T_{s_i}$ is restricted to one grid cell only of diameter $\alpha$.
If $e \not\in T_s$, then it must be that $e$ closes a cycle $C$ at the time it gets
processed. Note however that $C$ must lie entirely in the grid cell containing
$T_{s_i}$, since $C$ contains edges no longer than $\alpha$, and all edges with endpoints
in different cells are longer than $\alpha$.
Furthermore, $C$ must contain an edge $e' \not\in T_{s_i}$ such that $|e'| \le |e|$. The
case $|e'| = |e|$ cannot happen if Prim breaks ties in the same manner in both
$T_s$ and $T_{s_i}$, so it must be that $|e'| < |e|$.
But then we could replace $e$ in $T_{s_i}$ by $e'$, resulting in a smaller
spanning tree, a contradiction. It follows that $e \in T_s$ and therefore
$T_{s_i} \subseteq T_s$, for each $i$. This concludes the proof that
$\w(F_s) = O(\w(T_s))$, for each $s$. Since there are at most 9 such sets
$F_s$ that cover $F$ and since $\w(T_s) \le \w(MST)$, we get that $\w(F) = O(\w(MST))$.

It remains to prove that $\w(E_H \setminus F) = O(\w(MST))$.
Let $d \le 2k$ be the maximum number of edges in $E_H \setminus F$ incident to any node in $H$.
Partition the edge set $E_H \setminus F$ into no more than $2d \le 4k$ subsets
$E_1, E_2, \ldots$, such that no two edges in $E_i$ share a vertex, for each $i$.
We now show that $\w(E_i) = O(\w(MST))$, for each $i$. Since there are only a
constant number of sets $E_i$ ($4k$ at most), it follows that $\w(E_H \setminus F) = O(\w(MST))$.
The key observation to proving that $\w(E_i) = O(\w(MST))$ is that any two edges $uv, ab \in E_i$ have
their closest endpoints -- say, $u$ and $a$ --
separated by a distance of at least $r/t$. This is because
$t|ua| \ge |\ssp_H(u,a)| > r$;
the first part of this inequality follows from the spanner property of $H$, and
the second part follows from the fact that $u$ and $a$ are centers of different
$r$-clusters (a property ensured by step 4 of the algorithm). This implies that
$E_i$ is $r/t$-isolated, and by Lem.~\ref{lem:isolated} we have that
$\w(E_i) = O(\w(MST))$.

We have established that $\w(F) = O(\w(MST))$ and $\w(E_H\setminus F)=O(\w(MST))$.
It follows that $w(H) = w(E_H) = O(\w(MST))$ and this completes the proof.
\hfill\ABox
\end{proof}

\begin{theorem}
The \alg\ algorithm can be implemented in $O(1)$ rounds of communication using messages that
are $O(\log n)$ bits each.
\label{thm:communication}
\end{theorem}
\begin{proof}
Let $x_u$ and $y_u$ denote the coordinates of a node $u$. At the beginning of
the algorithm, each node $u$ broadcasts the information $(\id(u), x_u, y_u)$
to its neighbors and collects similar information from its neighbors.
Each node $u$ determines the grid cell(s) $(i,j)$ it belongs to from
two conditions, $i\alpha/\sqrt{2} \le  x_u < (i+1)\alpha/\sqrt{2}$
and
$j\alpha/\sqrt{2} \le  y_u < (j+1)\alpha/\sqrt{2}$.
Similarly, for each neighbor $v$ of $u$, each node $u$ determines the grid
cell(s) that $v$ belongs to. Thus step 1 of the algorithm can be implemented
in one round of communication: using the information from its neighbors,
each node $u$ computes the clique corresponding to those cells $(i,j)$ that
$u$ belongs to (at most 4 of them), then $u$ computes a $(1+\e)$-spanner for each
such clique by performing local computations.
Note that knowledge of node coordinates is critical to implementing step 1
efficiently.

Step 2 (the Yao step) and step 3 (the reverse Yao step) of the algorithm
are inherently local: each node $u$ computes its incident Yao and YaoYao
edges based on the information gathered from its neighbors in step 1.

It remains to show that step 4 can also be implemented in $O(1)$ rounds
of communication. We will in fact show that eight rounds of communication
suffice to compute an $r$-cluster cover for $V$ in $H$.
Define $U_s$ to be the set of vertices that lie in the grid cells $(i,j)$ such
that $(i \mod 2)\times 2+ j \mod 2 = s$. This is the same as saying that
two vertices that lie in different cells are about one grid cell apart.
Note that $V = \cup_{s=1}^4 U_s$. To compute an $r$-cluster cover for $V$,
each node $u$ executes the \cc\ method described below. For
simplicity we assume that $r > \delta$, so that two cluster centers that lie in different
grid cells are at least distance $r$ apart. However, the \cc\ method
can be easily extended to handle the situation $r \le \delta$ as well.

\begin{center}
\vspace{1mm}
\fbox{
\begin{minipage}[h]{0.8\linewidth}
\centerline{Computing a \cc($u, r$)}
\vspace{1mm}{\hrule width\linewidth}\vspace{2mm} 
\small{\begin{tabbing}
......\=......\=............\=.......\=..................................................\kill
Repeat for $s = 1,2,3,4$ \\
   \> (A) Collect information on cluster centers from neighbors (if any). \\ 
   \> If $u$ belongs to $U_s$ \\
   \> \> Let $V_\ell \subseteq U_s$ be the clique containing $u$ (computed in step 1 of \alg). \\
   \> \> (B) Broadcast information on existing cluster centers in $V_\ell$ to all nodes in $V_{\ell}$. \\
   \> \> (C) For each existing cluster center $w \in V_{\ell}$ \\
   \> \> \> Add to $C_w$ all uncovered nodes $v \in V_\ell$ such that $\ssp_H(w, v) \le r$. \\
   \> \> \> Mark all nodes in $C_w$ covered. \\
   \> \> (D) While $V_\ell$ contains uncovered nodes \\
   \> \> \> Pick the uncovered node $w \in V_{\ell}$ of highest $\id$. \\
   \> \> \> Add to $C_w$ all uncovered nodes $v \in V_\ell$ such that $\ssp_H(v, w) \le r$. \\
   \> \> \> Mark all nodes in $C_w$  covered. \\
   \> \> (E) Broadcast the cluster centers computed in step (C) to all neighbors.
\end{tabbing}}
\end{minipage}
}
\vspace{2mm}
\end{center}
No information on existing cluster centers is available in the first iteration of the \cc\ method (i.e, for $s = 1$). Each node in $U_1$ skips directly to step (D), which implements the standard greedy method for computed an $r$-clique cover for a given node set ($V_{\ell}$ in our case). In the second iteration, some of the clusters computed during the first iteration might be able to grow to incorporate new vertices from $U_2$. This is particularly true for cluster centers that lie in the overlap area of two neighboring cells. Information on such cluster centers is distributed to all relevant nodes in step (E) in the first iteration, then collected in step (A) and forwarded to all nodes in $V_{\ell}$ in step (B) in the second iteration. This guarantees that all nodes in $V_{\ell}$ have a consistent view of existing cluster centers in $V_\ell$ at the beginning of step (C). Existing clusters grow in step (C), if possible, and new clusters get created in step (D), if necessary. This procedure shows that it takes no more than 8 rounds of communication to implement step 4 of the \alg\ algorithm. One final note is that information on a \emph{constant} number of cluster centers is communicated among neighbors in steps (A), (B) and (D) of the \cc\ method. This is because only a constant number of $r$-clusters can be packed into a grid cell. So each message is $O(\log n)$ bits long,
necessarily so to include a constant number of node identifiers, each of which takes $O(\log n)$ bits.
\hfill\ABox
\end{proof}

\subsection{The \palg\ Algorithm}
\label{sec:alg2}
In this section we impose our spanner to be planar, at the expense of a
bigger stretch factor. This tradeoff is unavoidable, since there are UDGs
that contain no $(1+\e)$-spanner planar subgraphs, for arbitrarily small $\e$
(a simple example would be a square of unit diameter).

Our \palg\ algorithm consists of 4 steps. In a first step we construct
the unit Delaunay triangulation $\udel(V)$ using the method described
in~\cite{LCWW03}. Remaining steps use the grid-based idea from Sec.~\ref{sec:alg1}
to refine the Delaunay structure.
Let $V_1, V_2, \ldots$ be a $(\beta, \delta)$-clique cover for $V$, as
defined in Sec.~\ref{sec:alg1}.
In step 2 of the algorithm we apply the \oyao\ method
on edge subsets of $\udel$ incident to each clique $V_i$. The reason for restricting
this method to each clique, as opposed to the entire spanner $\udel(V)$ as
in~\cite{WangLi03}, is
to reduce the total of $O(n)$ rounds of communication to $O(1)$. The individual
degree of each node increases as a result of this alteration, however it remains
bounded above by a constant.
Steps 3 and 4 aim to reduce the total weight of the spanner. Step 3 uses
a Greedy method to filer out edges with both endpoints in one same clique $V_i$.
Step 4 uses clustering to filter out edges spanning multiple cliques.
These steps are described in detail in Table~\ref{tab:palg}.
\begin{table}[htpb]
\begin{center}
\fbox{
\begin{minipage}[ht]{0.98\linewidth}
\centerline{Algorithm \palg($G=(V,E), \e$)}
\vspace{1mm}{\hrule width\linewidth}\vspace{2mm} 
\small{\begin{tabbing}
.......\=......\=.......\=.........\=..........\=..........\=..............................\kill
{\bf \{1. Start with the localized Delaunay structure for $G$:\}} \\
\> Compute $\ldel = (V, E_{\ldel})$ for $G$ using the method from~\cite{LCWW03}. \\
\> Fix $0 < \beta \le \frac{1}{\sqrt{2}}$ and $0 < \delta < \frac{\beta}{4}$.
 Compute a ($\beta, \delta$)-clique cover $V_1, V_2, \ldots$ for $V$. \\
\\
{\bf \{2. Bound the degree:\}} \\
   \> For each clique $V_i$ do the following: \\
   \> \> 2.1 \> Let $E_i \subseteq E_{\udel}$ contain all unit Delaunay edges incident to nodes in $V_i$. \\
   \> \> 2.2 \> Execute $\ydel_i \leftarrow$ \oyao($G_i = (V, E_i)$) (see Table~\ref{tab:orderedyao}). \\
   \> Set $\ydel = (V, E_{\ydel}) = \bigcup_{i} \ydel_i$. \\
\\
{\bf \{3. Bound the weight of edges confined to single grid cells:\}} \\
\> Initialize $E_H = \emptyset$ and $H = (V, E_H)$. \\
\> Repeat for $k = 1,2,3,4$ \\ 
   \> \> {\bf \{Use Greedy on non-adjacent grid cells:\}} \\
\>   \> For each grid cell $\N=\N(i,j)$ such that $(i \mod 2)\times 2+ j \mod 2 = k$ \\
\>   \> 3.1 \> Let $E_{\N} = E_{\ydel} \cap \N$ contain all edges in \ydel\ that lie in $\N$. \\
\>   \> \> Let $E_{Q} = E_{\ydel} \setminus E_{\N}$ and $Q = (V, E_Q)$ define the query graph for $E_{\N}$. \\
\>   \> 3.2 \> Sort $E_{\N}$ in increasing order by edge \id. \\
\>   \> \> For each edge $e=uv \in E_{\N}$, resolve a shortest path query: \\
\>   \> \> \> If $\ssp_{Q}(u, v) > (1+\e)|uv|$ then add $uv$ to $H$ and $Q$.\\
\>   \> \> \> Otherwise, eliminate $uv$ from \ydel. \\
\\
{\bf \{4. Bound the weight of edges spanning multiple grid cells:\}} \\
\> Pick $r$ such that $r \le \frac{\e\delta}{4}$ and compute an $r$-cluster cover for $\ydel$. \\
\> Add to $H$ those edges in \ydel\ connecting cluster centers. \\
\\
{\bf Output $H = (V, E_H)$.}
\end{tabbing}}
\end{minipage}
}
\end{center}
\caption{The \palg\ algorithm.}
\label{tab:palg}
\end{table}
The reason for breaking up step 3 of the algorithm into 4 different rounds (for $k = 1,\ldots,4$) will become clear later, in our discussion of communication complexity (Thm.~\ref{thm:communication}). We now turn to proving some important properties of the output spanner. We start with a preliminary lemma.

\begin{lemma}
The graph \ydel\ constructed in step 2 of the \palg\ algorithm is a planar $t_1$-spanner for $G$, for any $t_1 > C_{del}(\frac{\pi}{2}+1)$. Furthermore, for each edge $ab \in G$, \ydel\ contains a $t_1$-spanner $ab$-path with all edges shorter than $ab$~\emph{\cite{WangLi03}}.
\label{lem:shortedges}
\end{lemma}
\begin{proof}
$\ldel$ is a planar $C_{del}$-spanner for $G$~\cite{LCWW03}. By Thm.~\ref{thm:oyao}, $\ydel_i$ is a planar $(\frac{\pi}{2}+1)$-spanner for $G_i$, for each $i$. These together with the fact that $\ldel = \bigcup_i G_i$ show that $\ydel$ is a $t_1$-spanner for $G$.

\begin{figure}[htbp]
\centering
\includegraphics[width=0.45\linewidth]{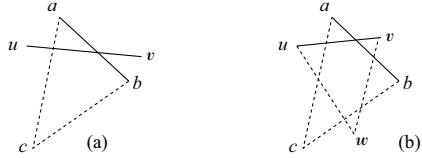}
\caption{\ydel\ is planar: edges $ab$ and $uv$ cannot cross.}
\label{fig:ydelplanar}
\end{figure}

The fact that $\ydel$ is planar follows an observation in~\cite{WangLi03} stating that, if a non-Delaunay edge $e \in \ydel$ crosses a Delaunay edge $e'$, then $e'$ must be longer than one unit and does not belong to \ydel. More precisely, the following properties hold:
\begin{itemize}
\item[(a)] A non-Delaunay edge $ab \in \ydel$ cannot cross a Delaunay edge $uv \in \ydel$. Recall that each non-Delaunay edge $ab \in \ydel$ closes an empty triangle $\triangle abc$ whose other two edges $ac$ and $bc$ are Delaunay edges. Thus, if $ab$ crosses $uv$, then at least one of $ac$ and $bc$ must cross $uv$, contradicting the planarity of \ldel (see Fig~\ref{fig:ydelplanar}a).

\item[(a)] No two non-Delaunay edges $ab, uv \in \ydel$ cross. The arguments here are similar to the ones above: if $ab$ and $uv$ intersect, then at least two of the incident Delaunay edges intersect, contradicting the planarity of \ldel (see Fig.~\ref{fig:ydelplanar}b).
\end{itemize}
The second part of the lemma follows from~\cite{WangLi03}. \hfill\ABox
\end{proof}

\begin{theorem}
The output $H$ generated by \palg($G, \e$) is a planar $t$-spanner for $G$, for any constant $t > C_{del}(1+\e)(1+\frac{\pi}{2})$.
\label{thm:pspanner}
\end{theorem}
\begin{proof}
Since $H \subseteq \ydel$, by Lem.~\ref{lem:shortedges} we have that $H$ is planar. We now show that $H$ is a $t$-spanner for $G$. The proof is by induction on the length of edges in $H$. The base case corresponds to the edge  $uv \in G$ of smallest \id. Clearly $uv \in \ldel$, since $uv$ is a Gabriel edge. Also $uv \in \ydel$, since it has the smallest \id\ among all edges and therefore it belongs to the Yao structure for \ldel. We now distinguish two cases:
\begin{itemize}
\item [(a)] There is a grid cell containing both $u$ and $v$. In this case $uv \in H$, since $uv$ is the first edge queried by Greedy in step 3 and therefore it gets added to $H$.
\item [(b)] There is no grid cell containing both $u$ and $v$. Let $ab$ be the edge selected in step 4 of the algorithm, such that $u \in \N_{a}$ and $v \in \N_{b}$ (see Fig.~\ref{fig:spanner}a). Then arguments similar to the ones used for the base case of Thm.~\ref{thm:spanner} show that $\ssp_H(u, a) \oplus ab \oplus \ssp_H(b, v)$ is a $(1+\e)$-spanner $uv$-path, for any  $r \le \e\delta/4$. 
\end{itemize}
This concludes the base case. To prove the inductive step, pick an arbitrary edge $uv \in G$, and assume that $H$ contains $t$-spanner paths between the endpoints of each edge in $G$ of smaller \id. By Lem.~\ref{lem:shortedges}, \ydel\ contains a $\frac{t}{1+\e}$-spanner path $u = u_0, u_1, \ldots, u_s = v$:
\begin{equation}
\sum_{i=0}^s |u_iu_{i+1}| \le \frac{t}{1+\e} |uv|
\label{eq:sp1}
\end{equation}
For each edge $u_iu_{i+1} \in \ydel$, one of the following cases applies:
\begin{itemize}
\item [(a)] There is a grid cell containing both $u_i$ and $u_{i+1}$. In this case, the Greedy step (step 3 of the algorithm) guarantees that $|\ssp_H(u_i,u_{i+1})| \le (1+\e)|u_iu_{i+1}|$.
\item [(b)] There is no grid cell containing both $u_i$ and $u_{i+1}$. Arguments similar to the ones for the base case show that $|\ssp_H(u_i, u_{i+1})| \le (1+\e)|u_iu_{i+1}|$.
\end{itemize}
In either case, $H$ contains a $(1+\e)$-spanner $u_iu_{i+1}$-path.
This together with~(\ref{eq:sp1}) shows that
$$|\ssp_H(u, v)| = \sum_{i=0}^s |\ssp_H(u_i, u_{i+1})| \le (1+\e) \sum_{i=0}^s |u_iu_{i+1}| \le t|uv|.$$
This completes the proof. \hfill\ABox
\end{proof}

\begin{theorem}
The output $H$ generated by \palg\ has maximum degree $O(1)$.
\label{thm:pdegree}
\end{theorem}
\begin{proof}
Since $H \subseteq \ydel$, it suffices to show that the graph \ydel\ constructed in step 2 of the \palg\ algorithm
has degree bounded above by a constant.
By Thm.~\ref{thm:oyao}, the maximum degree of $\ydel_i$ is 25, for each $i$. Also note that unit disk centered at a node $u$ intersects $O(\frac{1}{\beta^2})$ grid cells, meaning that $u$ is a neighbor of nodes in $O(\frac{1}{\beta^2})$ grid cells and therefore belongs to a constant number of graphs $\ydel_i$. This implies that the maximum degree of $u$ is $25 \cdot O(\frac{1}{\beta^2})$, which is a constant. \hfill\ABox
\end{proof}

\begin{definition}
{\bf [Leapfrog Property]}
For any $t \ge t' > 1$, a set $F$ of edges has the $(t',t)$-leapfrog property if,
for every subset $S = \{u_1v_1, u_2v_2, \ldots, u_mv_m\}$ of $F$,
\begin{equation}
t' \cdot |u_1 v_1| < \sum_{i=2}^m |u_i v_i| + t \cdot
\Big(\sum_{i=1}^{m-1} |v_i u_{i+1}| + |v_m u_1|\Big).
\label{eq:leapfrog}
\end{equation}
\end{definition}
Das and Narasimhan~\cite{DasNarasimhan97} show the following
connection between the leapfrog property and the weight of the
spanner.
\begin{lemma}
\label{lem:DasNarasimhan}
Let $~t \ge t' > 1$. If the line segments $F$ in $d$-dimensional
space satisfy the $(t', t)$-leapfrog property, then $\w(F) =
O(\w(MST))$, where $MST$ is a minimum spanning tree connecting the
endpoints of line segments in $F$.
\end{lemma}

\begin{lemma}
At the end of each iteration $k$ in step 3 of the \palg\ algorithm,
for $k = 1, \ldots, 4$, $Q$ contains $(1+\e)^k$-spanner paths between
the endpoints of any \ydel\ edge processed in iterations 1 through $k$.
\label{lem:step3greedy}
\end{lemma}
\begin{proof}
The proof is by induction on $k$. The base case corresponds to $k=1$.
In this case, Greedy ensures that $Q$ contains a $(1+\e)$-spanner $uv$-path for
each edge $uv$ processed in this iteration. This is because $uv \in \ydel$
either gets added to $Q$ in step 3.1 (and never removed thereafter),
or gets queried in step 3.2. To prove the inductive step, consider
a particular iteration $k > 1$, and assume that the lemma holds for iterations
$\ell = 1 \ldots k-1$.
Again Greedy ensures that $Q$ contains a $(1+\e)$-spanner $uv$-path for
each edge $uv$ processed in iteration $k$. Consider now an arbitrary edge
$uv$ processed in iteration $\ell < k$. By the inductive hypothesis, at
the end of round $k-1$, $Q$ contains a $(1+\e)^{k-1}$-spanner path $p(u,v)$.
However, it is possible that $p(u,v)$ contains edges processed in round
$k$ (since Greedy does not restrict $p(u,v)$ to lie entirely in the
cell containing $uv$). For each such edge, Greedy ensures the existence
of a $(1+\e)$-spanner path in $Q$. It follows that, at the end of iteration
$k$, $Q$ contains a $(1+\e)^k$-spanner $uv$-path. \hfill\ABox
\end{proof}

\begin{theorem} {\bf [Leapfrog Property]}
Let $L$ be an arbitrary grid cell and let $F \subseteq E_L$ be the set of edges with both
endpoints in $L$ that get added to $H$ in step 3 of
the algorithm.
Then $F$ satisfies the $(1+\e, t)$-leapfrog property, for $t = (1+\e)^4(\frac{\pi}{2}+1)C_{del}$.
\label{thm:H.leapfrog}
\end{theorem}
\begin{proof}
Consider an arbitrary subset $S = \{u_1 v_1, u_2 v_2, \ldots, u_m v_m\} \subseteq F$. To prove
inequality~(\ref{eq:leapfrog}) for $S$, it suffices to consider the
case when $u_1 v_1$ is a longest edge in $S$.
Define $S' = \{v_mu_1\} \cup\{v_{\ell} u_{\ell+1}~|~ 1 \le \ell < s\}$. Since
$u_i$ and $v_i$ lie in $L$ for each $i$, all edges from $S'$ lie entirely in $L$.
Let $ab \in S'$ be arbitrary. If $|ab| \ge |u_1v_1|$, then
inequality~(\ref{eq:leapfrog}) trivially holds, so assume that $|ab| < |u_1v_1|$.
Next we show that $Q$ contains an $ab$-path of length no greater than $t |ab|$
at the time $\{u_1, v_1\}$ gets queried.
We distinguish two cases:
\begin{itemize}
\item [(i)] $ab \in \ydel$. In this case $ab$ gets queried in step 3 prior to $u_1 v_1$,
meaning that $Q$ contains a path $P_Q(a,b)$ of length $|P_Q(a,b)| \le (1+\e)^4|ab|$, at the
time $u_1 v_1$ gets queried (by Lem.~\ref{lem:step3greedy}).
\item[(ii)] $ab \not\in \ydel$. By Lem.~\ref{lem:shortedges}, $\ydel$ contains a path $P_{\ydel}(a,b)$
of length
\begin{equation}
|P_{\ydel}(a,b)| \le \frac{t}{(1+\e)^4}|ab|
\label{eq:ydel}
\end{equation}
that contains only edges shorter than $ab$. For each edge $pq \in P_{\ydel}(a,b)$, $Q$ contains a path $P_Q(p,q)$ of length $|P_Q(p,q) \le (1+\e)^4|pq|$, at the time $u_1v_1$ gets queried (by Lem.~\ref{lem:step3greedy}). Thus we have that
\begin{equation}
|P_Q(a, b)| = \sum_{pq \in P_{\ydel}(a,b)} |P_Q(p,q)|
\le (1+\e)^4 \sum_{pq \in P_{\ydel}(a,b)} |pq| \le t |ab|
\end{equation}
This latter inequality follows from~(\ref{eq:ydel}).
\end{itemize}
For $1 \le k < s$, let $P_{\ell}$ be a shortest $v_{\ell}u_{\ell+1}$-path in $Q$, and let $P_m$ be a shortest $v_m u_1$-path in $Q$. By the arguments above, such paths exists in $Q$ at the time $u_1 v_1$ gets queried, and their stretch factor does not exceed $t$. Then $P = P_1 \oplus u_2v_2 \oplus P_2 \oplus u_3 v_3 \oplus \ldots \oplus P_m$ is a path from $u_1$ to $v_1$ in $Q$, and $\w(P)$ is no greater than the right hand side of the leapfrog inequality~(\ref{eq:leapfrog}).
Furthermore, $\w(P) > (1+\e) |u_1 v_1|$, otherwise the edge $u_1 v_1$ would not have been added to $H$ (and $Q$) in step 3 of the algorithm. This concludes the proof. \hfill\ABox
\end{proof}

\begin{theorem}
The output $H$ generated by \palg\ has total weight $O(\w(MST))$.
\label{thm:pweight}
\end{theorem}
\begin{proof}
The proof is very similar to the proof of Thm.~\ref{thm:degreeweight} and uses the results of
Lem.~\ref{lem:DasNarasimhan} and Thm.\ref{thm:H.leapfrog}. \hfill\ABox
\end{proof}

\begin{figure}[htbp]
\centering
\includegraphics[width=0.5\linewidth]{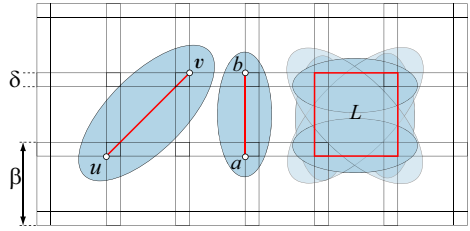}
\caption{Valid ranges for $|\ssp_H(u,v)| \le (1+\e)|uv|$ queries (step 3 of the \palg\ algorithm), illustrated for $\e = 1/2$: query range for edge $uv$ (left), for edge $ab$ (middle), and for the entire grid cell $L$ (right).}
\label{fig:sprange}
\end{figure}

\begin{lemma}
For any $\e < 2$, the shortest path query $|\ssp_{Q}(u,v)| \le (1+\e)|uv|$ in step 3 of the \palg\ algorithm involves
only those grid cells incident to the cell $\N$ containing $uv$.
\label{lem:spquery}
\end{lemma}
\begin{proof}
For a fixed edge $uv$, the locus of all points $z$ with the property that $|uz| + |zv| \le (1+\e)|uv|$ is a closed ellipse $A$ with focal points $u$ and $v$. Clearly, a point exterior to $A$ cannot belong to a $(1+\e)$-spanner path $p(u,v)$ from $u$ to $v$, so it suffices to limit the search for $p(u,v)$ to the interior of $A$. Fig.~\ref{fig:sprange} (left and middle) shows the search domains for edges corresponding to one diagonal ($uv$) and one side ($ab$) of a grid cell. For any grid cell $\N$, the union of $\N$ and the search ranges for the two diagonals and four sides of $\N$ covers the search domain for any edge that lies entirely in $\N$ (see Fig.~\ref{fig:sprange} right). It can be easily verified that, for $\e < 2$, the search domain for $\N$ fits in the union of $\N$ and its eight surrounding grid cells. \hfill\ABox
\end{proof}

\begin{theorem}
The \palg\ algorithm can be implemented in $O(1)$ rounds of communication.
\label{thm:communication}
\end{theorem}
\begin{proof}
Computing \ldel\ in step 1 of the algorithm takes at most 4 communication rounds~\cite{LCWW03}. As shown in the proof of Thm.~\ref{thm:degreeweight}, computing the clique cover in step 1 takes at most 8 rounds of communication.
Step 2 of the algorithm is restricted to cliques. A node $u$ belongs to at most 4 cliques. For each such clique, $u$ executes step 2 locally, on the neighborhood collected in step 1.
In a few rounds of communication, each node $u$ is also able to collect the information on the grid cells incident to the ones containing $u$. By Lem.~\ref{lem:spquery}, this information suffices to execute step 4 of the algorithm locally. \hfill\ABox
\end{proof}

\section{Conclusions}
We present the first localized algorithm that produces, for any given QUDG $G$ and any $\e > 0$, a
$(1+\e)$-spanner for $G$ of maximum degree $O(1)$ and total weight $O(\w(MST))$, in $O(1)$ rounds of communication.
We also present the first localized algorithm that produces, for any given UDG $G$, a planar
$O(1)$-spanner for $G$ of maximum degree $O(1)$ and total weight $O(\w(MST))$, in $O(1)$ rounds of communication.
Both algorithms require the use of a Global Positioning System (GPS), since each node uses its own
coordinates and the coordinates of its neighbors to take local decisions. Our work leaves open the
question of eliminating the GPS requirement without compromising the quality of the
resulting spanners.


\small

\def\cprime{$'$}
\begin{thebibliography}{10}

\bibitem{AraujoR04}
F.~Ara{\'u}jo and L.~Rodrigues.
\newblock Fast localized {D}elaunay triangulation.
\newblock In {\em OPODIS}, pages 81--93, 2004.

\bibitem{BDEK06}
P.~Bose, L.~Devroye, W.~Evans, and D.~Kirkpatrick.
\newblock On the spanning ratio of {G}abriel graphs and beta-skeletons.
\newblock {\em SIAM J. of Discrete Mathematics}, 20(2):412--427, 2006.

\bibitem{BGMS02}
P.~Bose, J.~Gudmundsson, and M.~Smid.
\newblock Constructing plane spanners of bounded degree and low weight.
\newblock In {\em {ESA} '02: Proc. of the 10th Annual European Symposium on
  Algorithms}, pages 234--246, London, UK, 2002. Springer-Verlag.

\bibitem{BMSU01}
P.~Bose, P.~Morin, I.~Stojmenovic, and J.~Urrutia.
\newblock Routing with guaranteed delivery in ad hoc wireless networks.
\newblock {\em Wireless Networks}, 7(6):609--616, 2001.

\bibitem{BRWZ04}
M.~Burkhart, P.~von Rickenbach, R.~Wattenhofer, and A.~Zollinger.
\newblock Does topology control reduce interference?
\newblock In {\em {Mobi{Hoc}} '04: 5th ACM Int. Symposium of Mobile Ad Hoc
  Networking and Computing}, pages 9--19, 2004.

\bibitem{d-yys-08}
M.~Damian.
\newblock A simple {Y}ao-{Y}ao-based spanner of bounded degree.
\newblock http://arxiv.org/abs/0802.4325v2, 2008.

\bibitem{dj-tcg-89}
G.~Das and D.~Joseph.
\newblock Which triangulations approximate the complete graph?
\newblock In {\em Proceedings of the international symposium on Optimal
  algorithms}, pages 168--192, New York, NY, USA, 1989. Springer-Verlag New
  York, Inc.

\bibitem{DN97}
G.~Das and G.~Narasimhan.
\newblock A fast algorithm for constructing sparse {E}uclidean spanners.
\newblock {\em Int. J. Comput. Geometry Appl.}, 7(4):297--315, 1997.

\bibitem{DasNarasimhan97}
G.~Das and G.~Narasimhan.
\newblock A fast algorithm for constructing sparse {E}uclidean spanners.
\newblock {\em Int. Journal on Computational Geometry and Applications},
  7(4):297--315, 1997.

\bibitem{Das95}
G.~Das, G.~Narasimhan, and J.~Salowe.
\newblock A new way to weigh malnourished {E}uclidean graphs.
\newblock In {\em SODA '95: Proceedings of the sixth annual ACM-SIAM symposium
  on Discrete algorithms}, pages 215--222, Philadelphia, PA, USA, 1995. Society
  for Industrial and Applied Mathematics.

\bibitem{dfs-dg-90}
D.~P. Dobkin, S.~J. Friedman, and K.~J. Supowit.
\newblock Delaunay graphs are almost as good as complete graphs.
\newblock {\em Discrete and Computational Geometry}, 5(4):399--407, 1990.

\bibitem{gg-geo-69}
K.R. Gabriel and R.R. Sokal.
\newblock A new statistical approach to geographic variation analysis.
\newblock {\em Systematic Zoology}, 18:259--278, 1969.

\bibitem{GGH+01}
J.~Gao, L.~Guibas, J.~Hershberger, L.~Zhang, and A.~Zhu.
\newblock Geometric spanner for routing in mobile networks.
\newblock In {\em Mobi{H}oc '01: Proc. of the 2nd {ACM} Int. Symposium on
  Mobile Ad Hoc Networking and Computing}, pages 45--55, 2001.

\bibitem{GK06}
J.~Gudmundsson and C.~Knauer.
\newblock Dilation and detours in geometric networks.
\newblock In T.F. Gonzalez, editor, {\em Handbook on Approximation Algorithms
  and Metaheuristics}, Boca Raton, 2006. Chapman \& Hall/CRC.

\bibitem{JC05}
T.~Johansson and L.~Carr-Moty\v{c}kov\'{a}.
\newblock Reducing interference in ad hoc networks through topology control.
\newblock In {\em DIALM-POMC '05: Proc. of the joint workshop on Foundations of
  mobile computing}, pages 17--23, 2005.

\bibitem{KarpKung00}
B.~Karp and H.~T. Kung.
\newblock Greedy perimeter stateless routing for wireless networks.
\newblock In {\em Mobi{C}om '00: Proc. of the 6th Annual ACM/IEEE Int.
  Conference on Mobile Computing and Networking}, pages 243--254, 2000.

\bibitem{kg-dt-92}
J.M. Keil and C.A. Gutwin.
\newblock The {D}elaunay triangulation closely approximates the complete
  {E}uclidean graph.
\newblock {\em Discrete and Computational Geometry}, 7:13--28, 1992.

\bibitem{li02sparse}
X.~Li, P.~Wan, Y.~Wang, and O.~Frieder.
\newblock Sparse power efficient topology for wireless networks.
\newblock In {\em HICSS'02: Proc. of the 35th Annual Hawaii Int. Conference on
  System Sciences}, volume~9, page 296.2, 2002.

\bibitem{li-localmst-03}
X.-Y. Li.
\newblock Approximate {MST} for {UDG} locally.
\newblock In {\em COCOON'03: Proc. of the 9th Annual Int. Conf. on Computing
  and Combinatorics}, pages 364--373, 2003.

\bibitem{LCW02}
X.~Y. Li, G.~Calinescu, and P.~Wan.
\newblock Distributed construction of planar spanner and routing for ad hoc
  wireless networks.
\newblock In {\em Info{C}om '02: Proc. of the 21st Annual Joint Conference of
  the IEEE Computer and Communications Societies}, volume~3, 2002.

\bibitem{LCWW03}
X.~Y. Li, G.~Calinescu, P.~J. Wan, and Y.~Wang.
\newblock Localized {D}elaunay triangulation with application in ad hoc
  wireless networks.
\newblock {\em IEEE Trans. on Parallel and Distributed Systems},
  14(10):1035--1047, 2003.

\bibitem{lws-ieee-04}
X.-Y. Li, Y.~Wang, and W.-Z. Song.
\newblock Applications of k-local mst for topology control and broadcasting in
  wireless ad hoc networks.
\newblock {\em IEEE Trans. on Parallel and Distributed Systems},
  15(12):1057--1069, 2004.

\bibitem{LiWang04}
X.Y. Li and Y.~Wang.
\newblock Efficient construction of low weight bounded degree planar spanner.
\newblock In {\em {COCOON} '04: Proc. of the 9th Int. Computing and
  Combinatorics Conference}, 2004.

\bibitem{lhs-infocom-03}
J.C.~Hou N.~Li and L.~Sha.
\newblock Design and analysis of a {MST}-based topology control algorithm.
\newblock In {\em IEEE INFOCOM}, 2003.

\bibitem{ns-gsn-07}
G.~Narasimhan and M.~Smid.
\newblock {\em Geometric Spanner Networks}.
\newblock Cambridge University Press, New York, NY, USA, 2007.

\bibitem{RajSurvey02}
R.~Rajaraman.
\newblock Topology control and routing in ad hoc networks: a survey.
\newblock {\em SIGACT News}, 33(2):60--73, 2002.

\bibitem{Smid00}
M.~Smid.
\newblock Closest-point problems in computational geometry.
\newblock In J.-R. Sack and J.~Urrutia, editors, {\em Handbook of Computational
  Geometry}, pages 877--935, Amsterdam, 2000. Elsevier Science.

\bibitem{SWLF04}
W.~Z. Song, Y.~Wang, X.~Y. Li, and O.~Frieder.
\newblock Localized algorithms for energy efficient topology in wireless ad hoc
  networks.
\newblock In {\em Mobi{H}oc '04: Proc. of the 5th ACM Int. Symposium on Mobile
  Ad Hoc Networking and Computing}, pages 98--108, 2004.

\bibitem{god-pr-80}
G.T. Toussaint.
\newblock The relative neighborhood graph of a finite planar set.
\newblock {\em Pattern Recognition}, 12(4):261--268, 1980.

\bibitem{RSWZ05}
P.~von Rickenbach, S.~Schmid, R.~Wattenhofer, and A.~Zollinger.
\newblock A robust interference model for wireless ad-hoc networks.
\newblock In {\em IPDPS '05: Proc. of the 19th IEEE Int. Parallel and
  Distributed Processing Symposium - workshop 12}, page 239.1, 2005.

\bibitem{MSTBroadcast02}
P.-J. Wan, G.~Calinescu, X.-Y. Li, and O.~Frieder.
\newblock Minimum energy broadcast routing in static ad hoc wireless networks.
\newblock {\em ACM Wireless Networking (WINET)}, 8(6):607--617, November 2002.

\bibitem{WangLi03}
Y.~Wang and X.~Y. Li.
\newblock Localized construction of bounded degree and planar spanner for
  wireless ad hoc networks.
\newblock In {\em Proc. of the Joint Workshop on Foundations of Mobile
  Computing}, pages 59--68, 2003.

\bibitem{Yao82}
A.C.-C. Yao.
\newblock On constructing minimum spanning trees in $k$-dimensional spaces and
  related problems.
\newblock {\em SIAM Journal on Computing}, 11(4):721--736, 1982.

\bibitem{ZK-qudg-07}
M.~Zuniga and B.~Krishnamachari.
\newblock An analysis of unreliability and asymmetry in low-power wireless
  links.
\newblock {\em ACM Trans. of Sensor Networks}, 3(2), 2007.

\end{thebibliography}
\def\cprime{$'$}

\end{document}